\documentclass[10pt,a4paper]{article}
\usepackage{fullpage}
\usepackage{amsfonts,amsmath,amssymb}
\usepackage{amsthm}
\usepackage{graphicx}\usepackage{sectsty}

\theoremstyle{plain}
\newtheorem{theorem}{Theorem}[section]
\newtheorem{lemma}[theorem]{Lemma}

\theoremstyle{definition}
\newtheorem{definition}[theorem]{Definition}
\theoremstyle{remark}
\newtheorem{remark}[theorem]{Remark}
\numberwithin{equation}{section}

\sectionfont{\large}

\begin{document}
\title{A Fixed Energy Fixed Angle Inverse Uniqueness in  Interior Transmission
Problem}
\author{Lung-Hui Chen$^1$}\maketitle\footnotetext[1]{Department of
Mathematics, National Chung Cheng University, 168 University Rd.
Min-Hsiung, Chia-Yi County 621, Taiwan. Email:
mr.lunghuichen@gmail.com;\,lhchen@math.ccu.edu.tw. Fax:
886-5-2720497.}
\begin{abstract}
We transform an inverse scattering problem to be an interior transmission problem.
We find an inverse uniqueness on the scatterer with a knowledge of a fixed interior transmission eigenvalue.
By examining the solution in a series of spherical harmonics at far fields, we can decide the perturbation uniquely for the radially symmetric perturbations.
\\MSC: 35P25/35R30/34B24.
\\Keywords: inverse problem/ inverse scattering/interior transmission eigenvalue/entire function theory/Rellich's lemma/Sturm-Liouville problem.
\end{abstract}

\section{Introduction}
In this work we study the inverse acoustic scattering problem of recovering the index of refraction in an inhomogeneous domain. To determine the inhomogeneity, we send a wave field into the domain. The propagation of the detecting wave field will be perturbed when hinging on the inhomogeneity defined by the index of refraction that in turn produces a scattered wave field. The inverse problem is to determine the index of refraction by the measurement of the scattered wave field. The study of inverse scattering problem is the core in various disciplines of science and technology such as sonar and radar, geophysical sciences, medical imaging, remote sensing, and non-destructive testing in instrument manufacturing.

\par
In this paper, we take the incident wave field to be the time harmonic  acoustic plane wave of the form $$u^i(x):=e^{ikx\cdot d},$$ $k\in\mathbb{R}^+$, $x\in\mathbb{R}^3$, and $d\in\mathbb{S}^2$ is the impinging direction.  The inhomogeneity is defined by the index of refraction $n\in\mathcal{C}^2(\mathbb{R}^3)$; $n(x)=n(|x|)>0$ and $n(|x|)-1\neq0$, for $x\in \Omega$, a ball of radius $R$ in $\mathbb{R}^3$.  The wave propagation is governed by the following equation.
\begin{eqnarray}\label{1.1}
\left\{%
\begin{array}{ll}
\Delta u(x)+k^2n(|x|)u(x)=0,\,x\in\mathbb{R}^3;\vspace{4pt}\\\vspace{3pt}
u(x)=u^i(x)+u^s(x),\,x\in\mathbb{R}^3\setminus \Omega; \\
\lim_{|x|\rightarrow\infty}|x|\{\frac{\partial u^s(x)}{\partial |x|}-iku^s(x)\}=0.
\end{array}%
\right.
\end{eqnarray}
Particularly, we have the following asymptotic expansion on the scattered wave field \cite{Colton2, Isakov}.
\begin{equation}\label{1.2}
u^s(x)=\frac{e^{ik|x|}}{|x|}u_\infty(\hat{x};d,k)+O(\frac{1}{|x|^{\frac{3}{2}}}),\,|x|\rightarrow\infty,
\end{equation}
which holds uniformly for all $\hat{x}:=\frac{x}{|x|}$, $x\in\mathbb{R}^3$, and $u_\infty(\hat{x};d,k)$ is known as the scattering amplitude in the literature \cite{Colton2,Hu,Isakov,Kirsch86,S}.
We state the following inverse uniqueness result in this paper.
\begin{theorem}\label{11}
Let $u_\infty^j(\hat{x};d,k)$, $j=1,2$, be the scattering amplitude parametered by  the non-constant index of refraction $n^j\in\mathcal{C}^2(\mathbb{R}^3)$ in~(\ref{1.1}).
If  $u_\infty^1(\hat{x};d,k)=u_\infty^2(\hat{x};d,k)$ for all $\hat{x}\in\mathbb{S}^2$ with a fixed $d\in\mathbb{S}^2$ and a fixed $k\geq 1$, then $n^1\equiv n^2$.
\end{theorem}
Let ${u^s}^j(x)$ be the scattered wave field parametered by $n^j$.
From Rellich's lemma in scattering theory \cite{Colton2,Hu,S},  the Theorem \ref{11} assumption and~(\ref{1.2}) imply
\begin{equation}\label{1.3}
{u^s}^1(x)={u^s}^2(x),\,x\in\mathbb{R}^3\setminus \Omega.
\end{equation}
Most importantly,~(\ref{1.1}),~(\ref{1.2}), and~(\ref{1.3}) imply the following interior transmission problem. Let us set $$w(x):={u^s}^1(x),\,v(x):={u^s}^2(x),$$ and thus
 we have the following system of equations.
\begin{eqnarray}\label{1.4}
\left\{%
\begin{array}{ll}
    \Delta w+k^2n^1w=0,  & \hbox{ in }\Omega; \vspace{3pt}\\\vspace{3pt}
    \Delta v+k^2n^2v=0, & \hbox{ in }\Omega; \\\vspace{3pt}
    w=v, & \hbox{ on }\partial \Omega; \\\vspace{3pt}
    \frac{\partial w}{\partial \nu}=\frac{\partial v}{\partial \nu},& \hbox{ on }\partial \Omega,\\
\end{array}%
\right.
\end{eqnarray}
where $\nu$ is the unit outer normal. The equation~(\ref{1.4}) is called the homogeneous interior transmission eigenvalue problem \cite{Aktosun,Cakoni,Cakoni2,Chen,Chen7,Colton,Colton2,Hu,Kirsch86,La,Mc,R,Sun2}.
The problem~(\ref{1.4}) occurs naturally when one considers the scattering of the plane waves by certain inhomogeneity inside the domain $\Omega$, defined by an index of refraction in many models.

\par
Now we expand the solution $(w,v)$ of~(\ref{1.4}) individually in two series of spherical harmonics by Rellich's lemma \cite[p.\,32]{Colton2}:
\begin{eqnarray}
&&w(x;k)=\frac{1}{r}\sum_{l=0}^{\infty}\sum_{m=-l}^{m=l}\alpha_{l,m}a_l(r;k)Y_l^m(\hat{x});\label{15}\\
&&v(x;k)=\frac{1}{r}\sum_{l=0}^{\infty}\sum_{m=-l}^{m=l}\beta_{l,m}b_l(r;k)Y_l^m(\hat{x}), \label{16}
\end{eqnarray}
where $r:=|x|$;  $\hat{x}=(\theta,\varphi)\in\mathbb{S}^2$, $k\in\mathbb{C}$.
The summations converge uniformly and absolutely on suitable compact subsets in $|r|\geq R_0$, with some sufficiently large $R_0\geq R$.
The expansion holds uniquely in the exterior domain of Helmholtz equation.
The spherical harmonics
\begin{equation}\label{Y}
\{Y_l^m(\theta,\varphi)\}_{l,m}:=\{\sqrt{\frac{2l+1}{4\pi}\frac{(l-|m|)!}{(l+|m|)!}}P_l^{|m|}(\cos\theta)e^{im\varphi}\}_{l,m},
\,m=-l,\ldots,l;\,l=0,1,2,\ldots,
\end{equation}
form a complete orthonormal system in $L^2(\mathbb{S}^2)$, in which
\begin{equation}\nonumber
P_n^m(t):=(1-t^2)^{m/2}\frac{d^mP_n(t)}{dt^m},\,m=0,1,\ldots,n,
\end{equation}
where the Legendre polynomials $P_n$, $n=0,1,\ldots,$ form a complete orthogonal system in $L^2[-1,1]$. By the analytic continuation of Helmholtz equation, the expansions~(\ref{15}) and~(\ref{16}) converge up to the boundary $\partial \Omega$, that is, $|x|=R$.

The orthogonality of the spherical harmonics \cite[p.\,227]{Colton2} implies that the functions in the form,
\begin{eqnarray}\label{17}
\left\{\begin{array}{ll}
&a_{l,m}(x;k):=\alpha_{l,m}\frac{a_l(r;k)}{r}Y_l^m(\hat{x}); \vspace{4pt}\\
&b_{l,m}(x;k):=\beta_{l,m}\frac{b_l(r;k)}{r}Y_l^m(\hat{x}),
  \end{array}
\right.
\end{eqnarray}
satisfy the first two equations in~(\ref{1.4}) independently for each $(l,m)$ for $|x|\geq R$. To fulfill the boundary condition of~(\ref{1.4}), we look for any $k\in\mathbb{C}$ such that
\begin{eqnarray}\label{18}
\left\{
  \begin{array}{ll}
    &\alpha_{l,m}a_l(r;k)|_{r=R}=\beta_{l,m}b_l(r;k)|_{r=R}; \vspace{6pt}\\
    &\alpha_{l,m}\partial_r\frac{a_l(r;k)}{r}|_{r=R}=\beta_{l,m}\partial_r\frac{b_l(r;k)}{r}|_{r=R}.
  \end{array}
\right.
\end{eqnarray}
By linear algebra, the existences of  $\alpha_{l,m}$ and $\beta_{l,m}$ in~(\ref{18}) are equivalent to finding the zeros of the following functional determinant:
\begin{eqnarray}\nonumber
D_{l}(k;R):=\det\left(%
\begin{array}{cc}
 a_l(r;k)|_{r=R}  & b_l(r;k)|_{r=R}\vspace{3pt}\\
  \{\frac{a_l(r;k)}{r}\}'|_{r=R}& \{\frac{b_l(r;k)}{r}\}'|_{r=R}
\end{array}%
\right)\label{110},
\end{eqnarray}
that is,
\begin{eqnarray}\label{DD}
D_l(k;R)
=\frac{a_l(R;k)b_l'(R;k)-a_l'(R;k)b_l(R;k)}{R}.
\end{eqnarray}

\par
Due to the radially symmetric assumption on~(\ref{1.1}), the Fourier coefficients $a_l(r;k)$ and $b_l(r;k)$  solve the following system of ODE for all $l\in\mathbb{N}_0$:
\begin{eqnarray}\label{111}
\left\{
  \begin{array}{ll}
    a_l''(r;k)+(k^2n^1(r\hat{x})-\frac{l(l+1)}{r^2})a_l(r;k)=0,\,0<r<\infty;\vspace{3pt}\\
    b_l''(r;k)+(k^2n^2(r\hat{x})-\frac{l(l+1)}{r^2})b_l(r;k)=0,\,0<r<\infty;\vspace{3pt}\\
    D_{l}(k;R)=0.
  \end{array}
\right.
\end{eqnarray}
By the uniqueness of the Sommerfeld radiation condition of~(\ref{1.1}), we deduce that
\begin{equation}\label{Q}
\alpha_{l,m}=\beta_{l,m}=1.
\end{equation}
We also set the initial conditions of $a_l(r;k)$ and $b_l(r;k)$ at $r=0$ to be the following conditions.
\begin{eqnarray}\label{113}
&&\lim_{r\rightarrow0}\{\frac{a_l(r;k)}{r}-j_l(k r)\}=0;\\
&&\lim_{r\rightarrow0}\{\frac{b_l(r;k)}{r}-j_l(k r)\}=0.\label{1.14}
\end{eqnarray}
The behavior of the Bessel function $j_l(kr)$ near $r=0$ is found in \cite[p.\,437]{Ab}. We refer the initial condition~(\ref{113}) and~(\ref{1.14}) to \cite{Mc}.
Considering~(\ref{113}) and~(\ref{1.14}) , we deduce~(\ref{111}) to be the following ODE.
\begin{eqnarray}\label{114}
\left\{
  \begin{array}{ll}
    a_l''(r;k)+(k^2n^1(r)-\frac{l(l+1)}{r^2})a_l(r;k)=0,\,0<r<\infty;\vspace{3pt}\\
    b_l''(r;k)+(k^2n^2(r)-\frac{l(l+1)}{r^2})b_l(r;k)=0,\,0<r<\infty;\vspace{3pt}\\
    D_{l}(k;0)=0;\vspace{3pt}\\
    D_{l}(k;R)=0.
  \end{array}
\right.
\end{eqnarray}
This is the reduced inverse problem~(\ref{1.1}) in radially symmetric form. Because $\{Y_l^m(\theta,\varphi)\}_{l,m}$ is complete in $L^2(\mathbb{S}^2)$, one element of the basis can be replaced by an element of fractional order of $(l,m)$. Hence, the property~(\ref{Y}) holds for $l\geq0$, $m\leq|l|$ without loss of generality.

\section{Asymptotic Expansions and Cartwright-Levinson Theory}
To estimate the asymptotic behaviors of the solution $a_l(r;k)$, $b_l(r;k)$, and then $D_l(k;R)$, we consider the following Liouville transformation:
\begin{eqnarray}\label{21}
&&z_l^1(\xi^1;k):=[n^1(r)]^{\frac{1}{4}}a_l(r;k),\mbox{ where } \xi^1:=\int_0^r[n^1(\rho)]^{\frac{1}{2}}d\rho;\\
&&z_l^2(\xi^2;k):=[n^2(r)]^{\frac{1}{4}}b_l(r;k),\mbox{ where } \xi^2:=\int_0^r[n^2(\rho)]^{\frac{1}{2}}d\rho.\label{22}
\end{eqnarray}
Let us set
\begin{equation}\nonumber
B^j:=\int_0^R[n^j(\rho)]^{\frac{1}{2}}d\rho,\,j=1,2.
\end{equation}

\par
If $k$ is an interior transmission eigenvalue, then
\begin{eqnarray}\label{23}
\left\{%
\begin{array}{ll}
    [z_l^j]''+[k^2-p^j(\xi^j)]z_l^j=0,\,0\leq\xi^j\leq B^j,\,j=1,2;\vspace{5pt}\\
   D_{l}(k;0)=0;\,D_{l}(k;R)=0,
\end{array}%
\right.
\end{eqnarray}
in which
\begin{equation}\label{24}
p^j(\xi^j):=\frac{[n^j]''(r)}{4[n^j(r)]^2}-\frac{5}{16}\frac{\{[n^j]'(r)\}^2}{[n^j(r)]^3}+\frac{l(l+1)}{r^2n^j(r)}.
\end{equation}
Here $\xi^j=B^j$ if and only if $r=R$.
Let
\begin{equation}
q^j(\xi^j):=\frac{[n^j]''(r)}{4[n^j(r)]^2}-\frac{5}{16}\frac{\{[n^j]'(r)\}^2}{[n^j(r)]^3}
+\frac{l(l+1)}{r^2n^j(r)}-\frac{l(l+1)}{[\xi^j]^2}.
\end{equation}
Thus,~(\ref{23}) and~(\ref{24}) imply that
\begin{eqnarray*}
\left\{%
\begin{array}{ll}
    [z_l^j]''+[k^2-q^j(\xi^j)-\frac{l(l+1)}{[\xi^j]^2}]z_l^j=0,\,0\leq\xi^j\leq B^j; \vspace{4pt}\\
    D_{l}(k;0)=0;\,D_{l}(k;R)=0.
\end{array}%
\right.
\end{eqnarray*}
\par
Let us drop the superscripts for simplicity. When $l=0$, the estimates for the solution $z_0(r)$ are classic and can be found in \cite{Po}:
\begin{equation}\label{p11}
z_0(\xi;k)=\frac{\sin k\xi}{k}-\frac{\cos
k\xi}{2k^2}Q(\xi)+\frac{\sin
k\xi}{4k^3}[p(\xi)+p(0)-\frac{1}{2}Q^2(\xi)]+O(\frac{\exp[|\Im
k|\xi]}{k^4}),
\end{equation}
where $Q(\xi):=\int_0^\xi p(s)ds$ and the error term has an improvement \cite[p.\,17]{Po} by
\begin{equation}\label{p12}
\frac{\cos{k\xi}}{8k^4}\{p'(\xi)-p'(0)-[p(\xi)+p(0)]Q(\xi)-\int_0^\xi p^2(s)ds+\frac{1}{6}Q^3(\xi)\}+o(\frac{\exp{|\Im k|\xi}}{|k|^4}).
\end{equation}
Similarly,
\begin{equation}\label{p13}
z'_0(\xi;k)= \cos k\xi+\frac{\sin k\xi}{2k}Q(\xi)+\frac{\cos
k\xi}{4k^2}[p(\xi)-p(0)-\frac{1}{2}Q^2(\xi)]+O(\frac{\exp[|\Im
k|\xi]}{k^3}),
\end{equation}
in which the boundary behavior of $p(\xi)$ plays a role in determining the inverse spectral uniqueness on the scatterer, and thus the $\mathcal{\mathcal{C}}^2$- assumption on the index of refraction is necessary.
\par
For $l\geq-\frac{1}{2}$, we apply the much more generalized results from \cite{Carlson,Carlson2}: Let $z_l(\xi,k)$ be the solution of
\begin{eqnarray}\label{Z}
\left\{\begin{array}{ll}
-z_l''(\xi)+\frac{l(l+1)z_l(\xi)}{\xi^2}+q(\xi)z_l(\xi)=k^2z_l(\xi);\vspace{6pt}\\
\lim_{\xi\rightarrow0}\frac{z_l(\xi)}{\xi^{l+1}}<\infty,
\end{array}
\right.
\end{eqnarray}
in which the function $q(\xi)$ is assumed to be real-valued and square-integrable. We note that the initial condition~(\ref{113}) and~(\ref{1.14}) imply the regularization condition at $\xi=0$ in~(\ref{Z}). The following asymptotics hold \cite[Lemma 3\,p.\,855]{Carlson3}.
\begin{eqnarray}\nonumber
|z_l(\xi;k)-\frac{\sin{k\xi}}{k}|\leq \frac{K\log(1+|k|)}{|k|^2}\exp\{|\Im k|\xi\};
\end{eqnarray}
\begin{eqnarray}\nonumber
|{z_l}'(\xi;k)-\cos{k\xi}|\leq \frac{K\log(1+|k|)}{|k|}\exp\{|\Im k|\xi\},\mbox{ where }K=K(\|q\|).
\end{eqnarray}
This explains the behaviors of solutions $z_l(\xi;k)$ for all $l\geq-\frac{1}{2}$.
\par
For the special case that $q(\xi)\equiv0$, we are actually considering the Bessel's equation:
\begin{eqnarray}
u_l''+[k^2-\frac{l(l+1)}{\xi^2}]u_l=0\label{129}.
\end{eqnarray}
The solutions of~(\ref{129}) essentially are Bessel's functions with a basis of two independent elements.

\par
The variation of parameters formula leads to the following pair of integral equations connecting $z_l(\xi,k)$, $u_l(\xi,k)$:
\begin{eqnarray}\nonumber
z_l(\xi,k)=u_l(\xi,k)-\int_0^\xi G(\xi,t,k)q(t)z_l(t,k,q)dt,
\end{eqnarray}
where
\begin{equation}\nonumber
G(\xi,t,k)=k^{-1}\Phi(kt,k\xi),
\end{equation}
in which
$\Phi(z,\omega)=\phi_1(z)\phi_2(\omega)-\phi_1(\omega)\phi_2(z)$ is defined as in \cite[p.\,6]{Carlson} that satisfies $\Phi(\omega,\omega)=0$ and
\begin{eqnarray}\nonumber
&&\phi_1(x)=(\frac{\pi x}{2})^{\frac{1}{2}}Y_{l+\frac{1}{2}}(x);\\
&&\phi_2(z)=(\frac{\pi z}{2})^{\frac{1}{2}}J_{l+\frac{1}{2}}(z),\nonumber
\end{eqnarray}
where $J_\nu(z)$ is the Bessel function of the first kind and $Y_\nu(x)$ is the Bessel function of the second kind. Similarly, we have the following integral equation.
\begin{eqnarray}\nonumber
z_l'(\xi;k)=u_l'(\xi;k)-\int_0^\xi H(\xi,t,k)q(t)z_l(t,k,q)dt,
\end{eqnarray}
\begin{equation}\nonumber
H(\xi,t,k)=k^{-1}\Psi(kt,k\xi),
\end{equation}
in which
$\Psi(z,\omega)=\phi_1(z)\phi_2'(\omega)-\phi_1'(\omega)\phi_2(z)$ with
is defined as in \cite[p.\,6]{Carlson} which satisfies $\Psi(\omega,\omega)=1$.
A pair of solutions of~(\ref{129}), are given by
\begin{eqnarray}
&&u_1(\xi;k)=k^l\phi_1(k\xi);\\
&&u_2(\xi;k)=k^{-(l+1)}\phi_2(k\xi).\label{210}
\end{eqnarray}
Moreover, we recall that
\begin{equation}
\label{215}
j_l(z)=\sqrt{\frac{\pi}{2z}}J_{l+\frac{1}{2}}(z),
\end{equation}
which we refer to \cite[p.\,437]{Ab}, wherein we find that $J_\nu(z)$ and $Y_\nu(z)$ are holomorphic functions of $z$ and entire functions of the order $\nu$ when $z$ is fixed. \textbf{In this paper, the complex analysis is focused at $\nu$.}

\par
For $\xi>0$ and $\Re k\geq0$, there is a constant $C$ such that
\begin{eqnarray}
&&|u_l(\xi;k)-\frac{\sin\{k\xi-l\frac{\pi}{2}\}}{k^{l+1}}|\leq C|k|^{-(l+1)}\frac{\exp\{|\Im k|\xi\}}{|k\xi|};\label{220}\\
&&|u_l'(\xi;k)-\frac{\cos\{k\xi-l\frac{\pi}{2}\}}{k^{l}}|\leq C|k|^{-l}\frac{\exp\{|\Im k|\xi\}}{|k\xi|}.\label{221}
\end{eqnarray}
Moreover,
\begin{eqnarray}
&&|z_l(\xi;k)-u_l(\xi;k)|\leq C(\frac{\xi}{1+|k\xi|})^{l+1}\exp\{|\Im k|\xi\}E(\xi;k);\label{222}\\
&&|z_l'(\xi;k)-u_l'(\xi;k)|\leq C(\frac{\xi}{1+|k\xi|})^{l}\exp\{|\Im k|\xi\}E(\xi;k),\label{223}
\end{eqnarray}
where
\begin{equation}\label{224}
E(\xi;k)=\exp\{\int_0^\xi\frac{tq(t)}{1+|kt|}dt\}-1.
\end{equation}
We refer these estimates to \cite[Lemma\,2.4, \,Lemma\,3.2]{Carlson2}.
These estimates give the fundamental asymptotic behaviors of $a_l(r;k)$, $a_l'(r;k)$, $b_l(r;k)$ and $b_l'(r;k)$ and, ultimately, the behavior of $D_{l}(k;R)$.
Most important of all, the estimates from~(\ref{220}) to~(\ref{223}) show that they are entire functions of order one and of type $\xi$. We refer details to \cite{Aktosun,Chen,Chen7,Colton2,Mc,Po}. Accordingly their zero sets are described by  Cartwright's theory \cite{Boas,Cartwright,Cartwright2,Levin,Levin2}. In particular, this leads to Weyl's type of asymptotics for the zeros of $D_{l}(k)$ which we describe as follows.
\begin{definition}
Let $f(z)$ be an integral function of order $\rho$, and let
$N(f,\alpha,\beta,r)$ denote the number of the zeros of $f(z)$
inside the angle $[\alpha,\beta]$ and $|z|\leq r$. We define the
density function of the zero set as
\begin{equation}\nonumber
\Delta_f(\alpha,\beta):=\lim_{r\rightarrow\infty}\frac{N(f,\alpha,\beta,r)}{r^{\rho}},
\end{equation}
and
\begin{equation}\nonumber
\Delta_f(\beta):=\Delta_f(\alpha_0,\beta),
\end{equation}
with some fixed $\alpha_0\notin E$ such that $E$ is at most a
countable set \cite{Boas,Koosis,Levin,Levin2}.
\end{definition}
\begin{theorem}\label{2222}
The functional determinant $D_l(k;R)$ is of order one and of type $B^1+B^2$. In particular,
\begin{equation}\nonumber
\Delta_{D_l(k;R)}(-\epsilon,\epsilon)=\frac{B^1+B^2}{\pi}.
\end{equation}
Similarly, $$\Delta_{a_l(k;R)}(-\epsilon,\epsilon)=\frac{B^1}{\pi}, \Delta_{a_l'(k;R)}(-\epsilon,\epsilon)=\frac{B^1}{\pi},$$
$$\Delta_{b_l(k;R)}(-\epsilon,\epsilon)=\frac{B^2}{\pi}, \Delta_{b_l'(k;R)}(-\epsilon,\epsilon)=\frac{B^2}{\pi}.$$
\end{theorem}
\begin{proof}
To sketch an idea of the proof, we note that the growth rate of $D_l(k;R)$ comes either from $a_l(R;k)b_l'(R;k)$ or $a_l'(R;k)b_l(R;k)$ by considering~(\ref{DD}). Either of them have growth rate  $B^1+B^2$ if we examine the estimates from~(\ref{220}) to~(\ref{223}).
These two terms do not cancel each other if $n^1\neq n^2$.
We skip the details and refer it to \cite{Chen,Chen3,Chen7,Colton7}.

\end{proof}
In particular, the spectrum of~(\ref{1.4}) is not empty and discrete in $\mathbb{C}$.

\section{Proof of Theorem \ref{11}}

Let $k$ be an eigenvalue of~(\ref{1.4}).  From the Rellich's expansion~(\ref{15}) and~(\ref{16}),~(\ref{18}) and then~(\ref{114}) hold for \textbf{all} $l\geq0$.
To prove Theorem \ref{11}, we consider an inverse Sturm-Liouville problem to~(\ref{114}) with a common eigenvalue holding for \textbf{all} $l\geq0$.

\par
We recall that $B^j$ and $z_l^j(\xi^j;k)$ are the quantities parametered by the index of refraction $n^j$, $j=1,2$, as in~(\ref{22}) and~(\ref{23}).
Therefore, the estimates~(\ref{220}),~(\ref{221}),~(\ref{222}),~(\ref{223}), and~(\ref{224}) hold for all $l\geq0$ and for both parities of indices of refraction. We have
\begin{equation}\label{31}
|u_l(B^j;k)-\frac{\sin\{kB^j-l\frac{\pi}{2}\}}{{k}^{l+1}}|\leq C|k|^{-(l+1)}\frac{\exp\{|\Im k |B^j\}}{|kB^j|},\,j=1,2,
\end{equation}
and
\begin{equation}\label{32}
|z_l^j(B^j;k)-u_l(B^j;k)|\leq C(\frac{B^j}{1+|kB^j|})^{l+1}\exp\{|\Im k|B^j\}E(B^j;k),
\end{equation}
in which $\exp\{|\Im k|B^j\}$ are bounded in a strip $S$ containing the real axis, and $E(B^j,k)$, $B^j\in\mathbb{R}^+$, is decreasing in $ k$. Firstly we apply~(\ref{31}) to obtain
\begin{equation}\label{33}
u_l(B^j;k)=\frac{\sin\{kB^j-\frac{l\pi}{2}\}}{{k}^{l+1}}[1+O(\frac{1}{|kB^j|})],\, l\in\mathbb{N}_0,
\end{equation}
in which $k$ is not in the zero set of $\sin\{kB^j-\frac{l\pi}{2}\}$ and the function is bounded near the real axis. Moreover, we deduce the following formula from~(\ref{210}) and~(\ref{215}).
\begin{equation}
k^{l+1}u_l(B^j;k)=(\frac{\pi k B^j}{2})^{\frac{1}{2}}J_{l+\frac{1}{2}}(kB^j)=kB^jj_l(kB^j).
\end{equation}
We note that a function of the form $zj_l(z)$ is called a Riccati-Bessel function \cite{Ab}.

\par
Furthermore,
\begin{eqnarray}\label{34}\nonumber
|{k}^{l+1}u_l(B^1;k)-{k}^{l+1}u_l(B^2;k)|&\leq&|{k}^{l+1}u_l(B^1;k)-{k}^{l+1}z_l^1(B^1;k)|\nonumber\vspace{4pt}
\\&&+|{k}^{l+1}z_l^1(B^1;k)-{k}^{l+1}z_l^2(B^2;k)|\nonumber
\\&&\vspace{4pt}+|{k}^{l+1}z_l^2(B^2;k)-{k}^{l+1}u_l(B^2;k)|.
\end{eqnarray}
From~(\ref{32}),
\begin{equation}
|{k}^{l+1}z_l^j(B^j;k)- {k}^{l+1}u_l(B^j;k)|\leq C(\frac{B^jk}{1+|kB^j|})^{l+1}e^{|\Im kB^j|}E(B^j;k)=o(1),\,\mbox{ as }l\rightarrow\infty,
\end{equation}
in which the estimate of $E(B^j;k)$ is given in~(\ref{224}).

\par

Secondly, an eigenvalue $k$ of~(\ref{1.4}) is surely an eigenvalue of~(\ref{114}) and~(\ref{18}). By Liouville's transformation~(\ref{21}),~(\ref{22}) with $n^1(R)=n^2(R)$, we deduce for this $k$ that 
\begin{equation}\label{3.4}
z_l^1(B^1;k)=z_l^2(B^2;k),\,\forall l\in\mathbb{R}^+.
\end{equation}
Assuming $k\geq1$ by Theorem \ref{11} assumption,~(\ref{34}), and~(\ref{3.4}) imply
\begin{equation}\nonumber
|{k}^{l+1}u_l(B^1;k)-{k}^{l+1}u_l(B^2;k)|\rightarrow0,\mbox{ as }l\rightarrow\infty,
\end{equation}
that is,
\begin{equation}\label{35}
|(kB^1)j_l(kB^1)-(kB^1)j_l(kB^2)|\rightarrow0,\mbox{ as }l\rightarrow\infty.
\end{equation}
\begin{lemma}
If~(\ref{35}) holds, then $B^1=B^2$.
\end{lemma}
\begin{proof}
From the Wronskian Identity \cite[p.\,439,\,(10.1.32)]{Ab}, we have
$$j_{n+1}(z)y_{n-1}(z)-j_{n-1}(z)y_{n+1}(z)=(2n+1)z^{-3}.$$
Because $y_n(z)=(-1)^{n+1}j_{-n-1}(z),\,n\in\mathbb{Z}$, we obtain that
\begin{equation}\label{311}
(-1)^nj_{n+1}(z)j_{-n}(z)-(-1)^{n}j_{n-1}(z)j_{-n-2}(z)=(2n+1)z^{-3}.
\end{equation}
Moreover,
$$Y_\nu(z)=\frac{J_\nu(z)\cos\{\nu\pi\}-J_{-\nu}(z)}{\sin\{ \nu\pi\}}.$$
Let us restrict $\nu\in\mathbb{Z}$. Thus, $J_{\nu}(z)=(-1)^\nu J_{-\nu}(z)$, and  $j_{\nu-\frac{1}{2}}(z)=(-1)^\nu j_{-\nu-\frac{1}{2}}(z)$ accordingly by~(\ref{215}).
Hence,~(\ref{35}) is deduced to be
\begin{equation}
|(kB^1)j_l(kB^1)-(kB^1)j_l(kB^2)|\rightarrow0,\mbox{ as }l\rightarrow\pm\infty\mbox{ in }\frac{1}{2}+\mathbb{Z}.
\end{equation}
In general, $j_\nu(z)$ is entire in $\nu$ and $|J_{\nu}(\Re z)|\leq1 $ for $\nu\geq0$ \cite[(9.1.60)]{Ab}. There are convergent subsequences for $\{(kB^1)j_l(kB^1)\}_{l\in\frac{1}{2}+\mathbb{Z}}$ and $\{(kB^2)j_l(kB^2)\}_{l\in\frac{1}{2}+\mathbb{Z}}$ with $kB^1$ and $kB^2$ fixed. Hence, the convergence holds for $l\in\mathbb{R}^+$, then for $l\in\mathbb{R}$.
Then, applying~(\ref{311}) for both parities imply that 
\begin{eqnarray*}
(-1)^nj_{n+1}(kB^i)j_{-n}(kB^i)-(-1)^{n}j_{n-1}(kB^i)j_{-n-2}(kB^i)=(2n+1)(kB^i)^{-3},\,i=1,2,
\end{eqnarray*}
that is, $|\frac{2l_m+1}{kB^1}-\frac{2l_m+1}{kB^2}|\rightarrow0$, as $m\rightarrow\infty$ for some subsequence $\{l_m\}\in\mathbb{N}$. This is not possible unless $B^1=B^2$.

\end{proof}
Applying this lemma, now we want to show that $z^1_0(B^1;k)\equiv z^2_0(B^1;k)$ in $\mathbb{C}$. In general,~(\ref{32}) implies that
\begin{eqnarray}\label{3.11}
z_l^j(B^1;k)=u_l(B^1;k)+O[(\frac{B^1}{1+|kB^1|})^{l+1}\exp\{|\Im k B^1|\}E(B^1;k)],\,j=1,2.
\end{eqnarray}
We apply~(\ref{33}) to~(\ref{3.11}), and obtain
\begin{eqnarray}
z_l^j(B^1;k)=\frac{\sin\{k B^1-\frac{l\pi}{2}\}}{k^{l+1}}[1+O(\frac{1}{k})],\,j=1,2.
\end{eqnarray}
Thus, we deduce the following asymptotic behavior of the quotient.
\begin{equation}\label{3.12}
\frac{z_l^1(B^1;k)}{z_l^2(B^1;k)}=\frac{1+O(\frac{1}{k})}{1+O(\frac{1}{k})}=1+O(\frac{1}{k}),\,k\in\mathbb{C},
\end{equation}
outside some neighborhoods of $\frac{\sin\{k B^1-\frac{l\pi}{2}\}}{k^{l+1}}$. Particularly, for $l=0$, we apply the asymptotics~(\ref{p11}) to~(\ref{3.12}). Let $\{\gamma_j\}_{j=1}^\infty$ be the simple zeros of $\frac{\sin\{k B^1\}}{k}$ and $\{\Gamma_j\}_{j=1}^\infty$ be some sequences of neighborhoods containing $\{\gamma_j\}_{j=1}^\infty$. We require that $|\Gamma_j|\rightarrow0$, as $j\rightarrow\infty.$

\par
For $k\notin\{\Gamma_j\}_{j=1}^\infty$, we note from~(\ref{3.12}) that
\begin{equation}\label{3.13}
\lim_{k\rightarrow\infty}\frac{z_0^1(B^1;k)}{z_0^2(B^1;k)}=1.
\end{equation}
For $k\in\{\Gamma_j\}_{j=1}^\infty$, we consider the following quotient.
\begin{equation}\label{3.14}
\lim_{j\rightarrow\infty}\lim_{k\rightarrow \gamma_j}\frac{z_0^1(B^1;k)}{z_0^2(B^1;k)}=\lim_{j\rightarrow\infty,\,|\Gamma_j|\neq0}\frac{\lim_{k\rightarrow k_j}\frac{z_0^1(B^1;k)}{\sin\{k B^1\}/k}}{\lim_{k\rightarrow k_j}\frac{z_0^2(B^1;k)}{\sin\{k B^1\}/k}}.
\end{equation}
From~(\ref{p11}) and~(\ref{p12}), we have the following power series.
\begin{equation}\label{3.15}
\frac{z_0^1(B^1;k)}{\frac{\sin\{kB^1\}}{k}}=1-\frac{\cot k B^1}{2k}Q(B^1)+\frac{[p(B^1)+p(0)-\frac{1}{2}Q^2(B^1)]}{4k^2}+O(\frac{1}{k^3}).
\end{equation}
Moreover, by complex analysis, we have
\begin{equation}\nonumber
\frac{z_0^1(B^1;k)}{\frac{\sin\{kB^1\}}{k}}
=\frac{{\rm Res}\{z_0^1(B^1;k)/\frac{\sin\{kB^1\}}{k};\gamma_j\}}{(k-\gamma_j)}+{\mbox{ higher order terms }},\, \mbox{ for }k\mbox{ near }\gamma_j.
\end{equation}
Now~(\ref{3.15}) implies that
\begin{eqnarray*}
&&{\rm Res}\{z_0^1(B^1;k)/\frac{\sin\{kB^1\}}{k};\gamma_j\}\\
&=&{\rm Res }\{1-\frac{\cot k B^1}{2k}Q^1(B^1)-\frac{1}{8k^2}[Q^1(B^1)]^2+\ldots;\gamma_j\}\\&=&
-\frac{Q^1(B^1)}{2\gamma_j}\cos\{\gamma_jB^1\}\times {\rm Res}\{\frac{1}{\sin\{kB^1\}};\gamma_j\}+O(\frac{1}{\gamma_j^2}).
\end{eqnarray*}
Note that  ${\rm Res}\{\frac{1}{\sin\{kB^1\}};\gamma_j\}=(-1)^j/B^1$.
Hence,
\begin{equation}\label{3.16}
{\rm Res}\{z_0^1(B^1;k)/\frac{\sin\{kB^1\}}{k};\gamma_j\}=-\frac{Q^1(B^1)}{2B^1\gamma_j}+O(\frac{1}{\gamma_j^2}).
\end{equation}
Let us examine $Q^1(B^1)$. By integration by parts, we deduce that
$$Q^1(B^1)=\int_0^{B^1}\frac{1}{4}\frac{{n^1}''}{{n^1}^2}-\frac{5}{16}\frac{{{n^1}'}^2}{{n^1}^3}ds
=\frac{3}{16}\int_0^{B^1}\frac{{{n^1}'}^2}{{n^1}^3}ds\geq0,$$
in which we emphasize that ${n^1}'(\xi^1=B^1)={n^1}'(r=R)=0$ by the fact that $n^1\in\mathcal{C}^2(\mathbb{R}^3)$ and is constant outside $\Omega$. Here, $Q^1(B^1)=0$ if and only if $n^1\equiv1$. The same analysis holds for the index $n^2$ as well.
Assuming $n^1$ and $n^2$  are non-constant as the assumption of Theorem \ref{11},~(\ref{3.14}) and~(\ref{3.16}), with non-vanishing second coefficient $Q^1(B^1)$ and $Q^2(B^1)$, imply that
\begin{equation}\label{3.18}
\lim_{j\rightarrow\infty}\lim_
{k\rightarrow \gamma_j}\frac{z_0^1(B^1;k)}{z_0^2(B^1;k)}=1.
\end{equation}
Hence,~(\ref{3.13}) and~(\ref{3.18}) imply that
\begin{equation}\label{3.21}
\lim_{k\rightarrow \infty}\frac{z_0^1(B^1;k)}{z_0^2(B^1;k)}=1.
\end{equation}
This implies that $\frac{z_0^1(B^1;k)}{z_0^2(B^1;k)}$ has only finite number of irreducible zeros, denoted as $\{z^1_1,z^1_2,\ldots,z^1_M\}$, or poles, $\{z^2_1,z^2_2,\ldots,z^2_M\}$, in $0i+\mathbb{R}$, in which we deduce from~(\ref{3.21}) that the numbers of the irreducible zeros and poles are equal to some $M\in\mathbb{N}_0$. Let $$F(k):=\frac{z_0^1(B^1;k)}{z_0^2(B^1;k)}.$$
Therefore,
\begin{equation}\label{FF}
F(k)=\frac{(k-z^1_1)(k-z^1_2)\cdots(k-z^1_M)}{(k-z^2_1)(k-z^2_2)\cdots(k-z^2_M)}.
\end{equation}
We note that whenever $\tilde{k}$ is an interior transmission eigenvalue of~(\ref{114}), then it satisfies $z_0^1(B^1;\tilde{k})=z_0^2(B^1;\tilde{k})$ by~(\ref{18}),~(\ref{Q}),~(\ref{21}), and~(\ref{22}).
That is $$F(\tilde{k})=1.$$
Theorem \ref{2222} implies that $$\frac{B^1+B^2}{\pi}>\frac{B^1}{\pi},\,\frac{B^1+B^2}{\pi}>\frac{B^2}{\pi}, $$
and then there is a higher density of zeros of interior transmission eigenvalues than the Dirichlet eigenvalues or Neumann eigenvalues of~(\ref{Z}).
Hence, we deduce from the Fundamental Theorem of Algebra to~(\ref{FF}) that
\begin{equation}\nonumber
F(k)\equiv1.
\end{equation}
In particular, $z_0^1(\xi;k)$ and $z_0^2(\xi;k)$ have the same Dirichlet eigenvalues.
\par
Similarly, we can prove that $[z_0^1]'(\xi;k)$ and $[z_0^2]'(\xi;k)$ have the same Neumann eigenvalues by considering~(\ref{p13})
and
\begin{equation}\nonumber
\frac{[z_0^1]'(B^1;k)}{[z_0^2]'(B^1;k)}=1+O(\frac{1}{k}),\,k\in\mathbb{C},
\end{equation}
outside some neighborhoods of the zeros of $\cos\{k B^1\}$. If $n^1$ and $n^2$ have the same set of Dirichet and Neumann eigenvalues, then  the inverse uniqueness result of the Bessel operator \cite[Theorem\,1.2,\,Theorem\,1.3]{Carlson2} implies that $n^1\equiv n^2$. This proves Theorem \ref{11}.

$\square$

\begin{remark}
It is believed that a stepwise potential function in general  can not be recovered by one discrete spectrum. See the Livshits' example on potentials in this class \cite{Zw}. In this paper,  a $\mathcal{C}^2$-index of refraction defines a series of ODE from far-fields to the origin.
\end{remark}



\end{document}